\documentclass[11pt,letterpaper]{article}

\usepackage{amsthm}
\newtheorem{theorem}{Theorem}
\newtheorem{lemma}{Lemma}

\newtheorem{corollary}{Corollary}

\theoremstyle{definition}
\newtheorem{definition}{Definition}

\newtheorem{example}{Example}
\usepackage[margin=1in]{geometry}

\usepackage{times}
\usepackage{soul}
\usepackage{url}
\usepackage[hidelinks]{hyperref}
\usepackage[utf8]{inputenc}
\usepackage[small]{caption}
\usepackage{graphicx}
\usepackage{amsmath,amsfonts,amssymb}
\usepackage{booktabs}
\usepackage{algorithm}
\usepackage{algorithmic}
\usepackage{xcolor}
\usepackage{xspace}
\usepackage{cleveref}
\usepackage{nicefrac}
\let\oldtextbf\textbf
\renewcommand{\textbf}[1]{\oldtextbf{\boldmath #1}}

\usepackage[square,sort,semicolon,numbers]{natbib}
\usepackage{amsfonts}

\renewcommand{\succeq}{\succcurlyeq}
\renewcommand{\ge}{\geqslant}
\renewcommand{\geq}{\geqslant}
\renewcommand{\vec}{\mathbf}

\newcommand{\N}{\mathbb{N}}
\newcommand{\R}{\mathbb{R}}
\newcommand{\set}[1]{{\left\{#1\right\}}}
\DeclareMathOperator{\argmax}{argmax}
\DeclareMathOperator{\argmin}{argmin}
\newcommand{\floor}[1]{\lfloor #1 \rfloor}
\newcommand{\ceil}[1]{\lceil #1 \rceil}

\newcommand{\MNW}{\operatorname{MNW}}
\newcommand{\mnw}{\operatorname{mnw}}
\newcommand{\lex}{\operatorname{lex}}
\newcommand{\tie}{\operatorname{tie}}
\newcommand{\round}{\operatorname{round}}
\newcommand{\alloc}{\vec{A}}
\newcommand{\dist}{\overline{\alloc}}
\newcommand{\allallocs}{\mathbb{A}}

\newcommand{\agents}{{\mathcal{N}}}
\newcommand{\goods}{{\mathcal{M}}}

\newcommand{\vals}{{\vec{v}}}
\newcommand{\mnwt}{\ensuremath{\MNW^{\text{tie}}}}
\newcommand{\si}{S^i_\alloc}
\newcommand{\li}{L^i_\alloc}
\newcommand{\alltruth}{\alloc^{\!\text{truth}}}
\newcommand{\alllie}{\alloc^{\!\text{lie}}}
\newcommand{\ati}[1]{A^{\text{truth}}_{#1}}
\newcommand{\ali}[1]{A^{\text{lie}}_{#1}}
\newcommand{\sit}{S^i_{\alltruth}}
\newcommand{\nonvalued}{{\mathcal{Z}}}

\newcommand{\uvec}{\vec{u}}
\newcommand{\umnw}{\vec{u}^{\mnw}}
\newcommand{\ulex}{\vec{u}^{\lex}}
\newcommand{\utie}{\vec{u}^{\tie}}
\newcommand{\wlex}{\vec{w}^{\lex}}
\newcommand{\wmnw}{\vec{w}^{\mnw}}

\newcommand{\utiek}[1]{u^{\tie}_{#1}}
\newcommand{\uround}{\vec{u}^{\round}}

\newcommand{\sorted}{\vec{s}}
\newcommand{\allocs}{\alloc^{\!*}}
\newcommand{\us}{\vec{u}^{\!*}}
\newcommand{\allocround}{\alloc^{\round}}
\newcommand{\atie}{\alloc^{\tie}}
\newcommand{\atiek}[1]{A^{\tie}_{#1}}
\newcommand{\cE}{\mathcal{E}}

\newcommand{\emdash}{\,---\,}

\title{Fair Division with Binary Valuations:\\One Rule to Rule Them All}
\author{
 Daniel Halpern\\University of Toronto\\\texttt{daniel.halpern@mail.utoronto.ca}
 \and
 Ariel D. Procaccia\\Harvard University\\\texttt{arielpro@seas.harvard.edu}
 \and
 Alexandros Psomas\\Google Research\\\texttt{alexpsomi@cs.berkeley.edu}
 \and
 Nisarg Shah\\University of Toronto\\\texttt{nisarg@cs.toronto.edu}
}
\date{}

\begin{document}
\maketitle
\begin{abstract}
	We study fair allocation of indivisible goods among agents. Prior research focuses on additive agent preferences, which leads to an impossibility when seeking truthfulness, fairness, and efficiency. We show that when agents have binary additive preferences, a compelling rule\emdash maximum Nash welfare (MNW)\emdash provides all three guarantees.
	
	Specifically, we show that deterministic MNW with lexicographic tie-breaking is group strategyproof in addition to being envy-free up to one good and Pareto optimal. 
	We also prove that fractional MNW\emdash known to be group strategyproof, envy-free, and Pareto optimal\emdash can be implemented as a distribution over deterministic MNW allocations, which are envy-free up to one good. Our work establishes maximum Nash welfare as the ultimate allocation rule in the realm of binary additive preferences.
\end{abstract}


\section{Introduction}

Fair division~\cite{Moul03,BT96} is a sprawling field that cuts across scientific disciplines. Among its many challenges, the division of indivisible goods\emdash an ostensible oxymoron\emdash is arguably the most popular in recent years. The goods are ``indivisible'' in the sense that each must be allocated in its entirety to a single agent (think of pieces of jewelry or tickets to different football games in a season). Each agent has her own \emph{valuation function}, which represents the benefit the agent derives from bundles of goods.  

A fully expressive model of valuation functions would have to account for combinatorial preferences. Classic examples include a right shoe that is worthless without its matching left shoe (complementarities), and two identical refrigerators (substitutes). However, rich preferences can be difficult to elicit. It is often assumed, therefore, that the valuation functions are \emph{additive}, that is, that each agent's value for a bundle of goods is the sum of her values for individual goods in the bundle. Additive valuations strike a balance between expressiveness and ease of elicitation; in particular, each agent need only report her value for each good separately. 

Another advantage of additive valuations is that they admit a practical rule that is both (economically) efficient and fair. Specifically, the Maximum Nash Welfare (MNW) solution\emdash which maximizes the product of valuations and, therefore, is obviously Pareto optimal (PO)\emdash is envy-free up to one good (EF1): for any two agents $i$ and $j$, it is always the case that $i$ prefers her own bundle to that of $j$, possibly after removing a single good from the latter bundle~\cite{CKMP+19}. 

The MNW solution, however, is not \emph{strategyproof}, that is, agents can benefit by misreporting their preferences. In fact, under additive valuations, the only Pareto optimal and strategyproof rule is \emph{serial dictatorship}, which is patently unfair~\cite{KM01}. This profound clash between efficiency and truthfulness holds true even when agents can only have three possible values for goods! 

The only hope for reconciling efficiency, fairness and truthfulness, therefore, is to assume that agents' values for goods are \emph{binary}. This assumption is not just a theoretical curiosity: while it obviously comes at a significant cost to expressiveness, it leads to extremely simple elicitation. In this sense, it arguably represents another natural point on the conceptual expressiveness-elicitation Pareto frontier. The same bold tradeoff has long been considered sensible in the literature on voting, where binary values are implicitly represented as \emph{approval} votes~\cite{BF07}; in fact, the assumption underlying some of the recent work on approval-based multi-winner elections~\cite{CJMW19,LS19} is nothing but that of binary additive valuations. It is not surprising, therefore, that several papers in fair division pay special attention to the case of binary additive valuations~\cite{BL16,DS15b,BKV18b,AAGW15,FSVX19}. 

With this rather detailed justification for binary additive valuations in mind, our primary research question is this: \emph{do binary additive valuations admit rules that are efficient, fair, and truthful?}

\subsection{Our Contribution}\label{sec:contrib}
We provide a positive answer\emdash and then some. Specifically, \Cref{thm:det-mnw} asserts that, under binary additive valuations, a particular form of the MNW solution is Pareto optimal, EF1, \emph{group} strategyproof (even a coalition of agents cannot misreport its members' preferences in a way that benefits them all) and polynomial-time computable. 

Furthermore, we show (\Cref{thm:rand-leximin-mnw-properties,thm:rand}) that by randomizing over MNW allocations, we can achieve ex ante envy-freeness (each agent's expected value for their random allocation is at least as high as for any other agent's), ex ante Pareto optimality, ex ante group strategyproofness, and ex post EF1 simultaneously in polynomial time. In other words, randomization allows us to circumvent the mild unfairness that is inherent in deterministic allocations of indivisible goods without losing the other guarantees.

In our view, these results are essentially the final word on how to divide indivisible goods under binary additive valuations.

\subsection{Related Work}\label{sec:related}

There is an extensive body of work on fair division, much too large to survey here. Instead, we focus on the most closely related work on fair division with \emph{binary} valuations.

The most closely related work is that of \citet{babaioff2020fair}, who, independently and in parallel to our work, also discovered some of the results that we present for the deterministic MNW rule. Specifically, their \textit{prioritized egalitarian} mechanism is identical to our deterministic $\mnwt$ mechanism presented in \Cref{sec:determ}. They show that this rule is strategyproof, EFX,\footnote{There are two popular definitions of EFX (see~\cite{ABFH+20}); this result holds for the stronger one: an allocation is EFX if the envy that one agent has toward another can be eliminated by removing \emph{any} good from the envied agent's bundle.} PO, Lorenz-dominating, and polynomial-time computable. This is very similar to our \Cref{thm:det-mnw}. The difference is that we strengthen strategyproofness to group strategyproofness, but only establish EF1 (weaker than EFX) and do not establish Lorenz-dominance. We note that the EFX property is also established by \citet{ABFH+20}. We view these results as complementary to ours, and together, they establish that $\mnwt$ is group strategyproof, EFX, PO, Lorenz-dominating, and polynomial-time computable, making it even more compelling. We note that \citet{babaioff2020fair} do not study randomized allocation rules, which we focus on in \Cref{sec:randomized}. 

\citet{ortega2020multi} studies a slightly more general problem where there may be multiple copies of each good, but each agent can receive at most one copy of any good. His \textit{egalitarian solution} is identical to our fractional MNW rule in terms of the probability of each good going to each agent, but he does not discuss how to implement these fractional allocations as a distribution over integral allocations with good properties. He shows that this rule is ex ante envy-free, ex ante PO, and ex ante group strategyproof. However, he uses a weaker notion of strategyproofness, where agents are only allowed to report a good that they like as one that they do not like, but not vice-versa. As we note in \Cref{sec:randomized}, in our (standard) setting with a single copy of each good, these guarantees (including the stronger strategyproofness notion, or even group strategyproofness) follow directly from prior work~\cite{KPS18}. Hence, our main focus in \Cref{sec:randomized} is to prove an \emph{ex post} EF1 guarantee, which \citet{ortega2020multi} does not provide. 

Two central concepts in our work are those of maximum Nash welfare (MNW) and leximin allocations. \citet{aziz2019almost} show that under binary additive valuations, all leximin allocations are also MNW allocations. As we observe in \Cref{sec:determ}, this, together with known properties of the two solutions, immediately implies that the sets of MNW and leximin allocations are identical. \citet{benabbou2020finding} extend this equivalence to a more general valuation class. 

On the computation front, our polynomial-time computability result for the deterministic $\mnwt$ rule builds upon on efficient algorithms by \citet{DS15b} and \citet{BKV18b} for finding an MNW allocation under binary additive valuations; specifically, our algorithm starts from an arbitrary MNW allocation computed by either of these algorithms, and then iteratively finds a \emph{special} MNW allocation that $\mnwt$ outputs. \citet{benabbou2020finding} also show that an MNW allocation can be computed efficiently under their more general valuation class. 


\section{Preliminaries}\label{sec:prelim}

For $k \in \N$, let $[k] = \{1,\ldots,k\}$. Let $\agents = [n]$ denote a set of \emph{agents}, and $\goods$ denote a set of $m$ indivisible \emph{goods}. Each agent $i$ is endowed with a \emph{valuation} function $v_i: 2^\goods \rightarrow \R_{\geq 0}$ such that $v_i(\emptyset) = 0$. It is assumed that valuations are additive: $\forall T \subseteq \goods$, $v_i(T) = \sum_{g \in T} v_i(\set{g})$. To simplify notation, we write $v_i(g)$ instead of $v_i(\set{g})$. 

We focus on a subclass of additive valuations known as binary additive valuations, under which $v_i(g) \in \set{0,1}$ for all $i \in \agents$ and $g \in \goods$. We say that agent $i$ \emph{likes} good $g$ if $v_i(g) = 1$. Sometimes it is easier to think of the valuation function of agent $i$ as the set of goods that agent $i$ likes, denoted $V_i = \set{g \in \goods : v_i(g) = 1}$. Note that $v_i(T) = |V_i \cap T|$ for all $T \subseteq \goods$. For a set of agents $S \subseteq \agents$, let $V_S = \bigcup_{i \in S} V_i$ be the set of goods that at least one agent in $S$ likes. The vector of agent valuations $\vals = (v_1,\ldots,v_n)$ is called the \emph{valuation profile}. A problem instance is given by the tuple $(\agents, \goods, \vals)$.

For a set of goods $T \subseteq \goods$ and $k \in \N$, let $\Pi_k(T)$ denote the set of partitions of $T$ into $k$ bundles. We say that $\alloc = (A_1,\ldots,A_n)$ is an allocation if $\alloc \in \Pi_n(T)$ for some $T \subseteq \goods$. Here, $A_i$ is the bundle of goods allocated to agent $i$, and $v_i(A_i)$ is the \emph{utility} to agent $i$. Let us denote $A_S = \bigcup_{i \in S} A_i$ for $S \subseteq \agents$. Let $\allallocs = \bigcup_{T \subseteq \goods} \Pi_n(T)$ denote the set of all allocations. 

We say that good $g$ is \emph{non-valued} if $v_i(g) = 0$ for all agents $i$; all the remaining goods are called \emph{valued}. Let $\nonvalued$ denote the set of non-valued goods. We say that an allocation $\alloc$ is \emph{complete} if it allocates every valued good, i.e., if $A_{\agents} \supseteq \goods \setminus \nonvalued$; we say that it is \emph{minimally complete} if it is complete and does not allocate any non-valued goods, i.e., if $A_{\agents} = \goods \setminus \nonvalued$.

We are interested in \emph{fair} allocations. One of the most prominent notions of fairness is envy-freeness~\cite{Fol67}.
\begin{definition}[Envy-freeness]
    An allocation $\alloc$
    is called \emph{envy-free} (EF) if, for all agents $i,j \in \agents$, $v_i(A_i) \geq v_i(A_j)$.
\end{definition}

Envy-freeness requires that no agent prefer another agent's bundle over her own. This cannot be guaranteed (imagine two agents liking a single good). Prior literature focuses on its relaxations, such as envy-freeness up to one good~\cite{LMMS04,Bud11}, which can be guaranteed. 

\begin{definition}[Envy-freeness up to one good]
    An allocation $\alloc$ is called \emph{envy-free up to one good} (EF1) if, for all agents $i,j \in \agents$ such that $A_j\neq \emptyset$, there exists $g \in A_j$ such that $v_i(A_i) \ge v_i(A_j\setminus\set{g})$. 
\end{definition}

EF1 requires that it should be possible to remove envy between any two agents by removing at most one good from the envied agent's bundle. We remark that there is a stronger fairness notion called envy-freeness up to the least positively valued good (EFX)~\cite{CKMP+19}, which coincides with EF1 under binary additive valuations.\footnote{There are two popular definitions of EFX~\cite{ABFH+20}. The original definition by \citet{CKMP+19} asks that agent $i$ not envy agent $j$ after removal of any good from agent $j$'s bundle that has \emph{positive} value for agent $i$, whereas a latter definition omits the requirement of ``positive value''. Under binary additive valuations, the former definition is equivalent to EF1 whereas the latter definition is stronger than EF1.}

Another classic desideratum in resource allocation is Pareto optimality, which is a notion of economic efficiency. 

\begin{definition}[Pareto optimality]
    An allocation $\alloc$ is called \emph{Pareto optimal} (PO) if there does not exist an allocation $\alloc'$ such that for all agents $i \in \agents$, $v_i(A'_i) \geq v_i(A_i)$, and at least one inequality is strict. 
\end{definition}

It is easy to see that with binary additive valuations, Pareto optimality is equivalent to ensuring that each valued good is allocated to one of the agents who likes it, i.e., that the utilitarian social welfare (sum of utilities) is maximized and is equal to the number of valued goods. 


\section{Deterministic Setting}\label{sec:determ}

In this section, our main goal is to establish the existence of a deterministic allocation rule that is fair, efficient, and truthful under binary additive valuations. Our rule builds upon the concept of maximum Nash welfare allocations~\cite{CKMP+19}, which we define below. 

\begin{definition}[Maximum Nash welfare allocation]
	We say that $\alloc$ is a \emph{maximum Nash welfare} (MNW) allocation if, among the set of allocations $\allallocs$, it maximizes the number of agents receiving positive utility and, subject to that, maximizes the product of positive utilities. Formally, let $W(\alloc) = \set{i \in \agents : v_i(A_i) > 0}$ and $\allallocs_M = \argmax_{\alloc \in \allallocs} |W(\alloc)|$. Then, $\argmax_{\alloc \in \allallocs_M} \prod_{i \in W(\alloc)} v_i(A_i)$ is the set of MNW allocations.
\end{definition}

Even under general additive valuations, all maximum Nash welfare allocations satisfy EF1 and PO~\cite{CKMP+19}. Our work uses a connection between MNW allocations and the classic concept of leximin allocations, that holds under binary additive valuations. 

\begin{definition}[Leximin comparison]
	For an allocation $\alloc$, let its \emph{utility vector} be $(v_1(A_1),\ldots,v_n(A_n))$, and its \emph{utility profile} be the utility vector sorted in a non-descending order. Given two utility profiles $\sorted = (s_1,\ldots,s_n)$ and $\sorted' = (s'_1,\ldots,s'_n)$, we say that $\sorted$ leximin-dominates $\sorted'$, denoted $\sorted \succ_{\lex} \sorted'$, if there exists $k \in [n]$ such that $u_k > u'_k$ and $u_r = u'_r$ for all $r < k$. We say that $\sorted$ weakly leximin-dominates $\sorted'$, denoted $\sorted \succeq_{\lex} \sorted'$, if $\sorted \succ_{\lex} \sorted'$ or $\sorted = \sorted'$. Note that this is a total order among utility profiles. We extend these comparisons to utility vectors by applying them to the utility profiles they induce, and call two utility vectors leximin-equivalent if they induce the same utility profile. 
\end{definition}

\begin{definition}[Leximin allocations]
	We say that $\alloc$ is a \emph{leximin allocation} if, among all allocations, it lexicographically maximizes the utility profile, i.e., maximizes the minimum utility, subject to that maximizes the second minimum, and so on. Thus, leximin allocations are those whose utility profile is the greatest element of the total order $\succ_{\lex}$. We also extend the notions of leximin-dominance and weak leximin-dominance to allocations by comparing their utility vectors.
\end{definition}

Leximin is a refinement of the traditional Rawlsian fairness, which requires maximization of the minimum utility. \citet{PR18} and \citet{FSVX19} study leximin allocations (and variants of this definition), and show that they have related fairness properties as well. 

Important to our work is the observation that for binary additive valuations, the sets of leximin and MNW allocations coincide. This is established under a more general valuation class by the contemporary work of \citet{benabbou2020finding}, but for binary additive valuations, this can also be inferred easily from the following observations, which we will use in our work. 

\begin{lemma}\label{lem:leximin-unique-profile}
	All leximin allocations have the same utility profile. Further, any allocation with this utility profile is a leximin allocation.
\end{lemma}
\begin{proof}
This is because lexicographic comparison is a total order among utility profiles, and leximin allocations, by definition, are those whose utility profile is its greatest element. 
\end{proof}

\begin{lemma}[Lemma~21 of \citet{FSVX19}]\label{lem:nash-unique-profile}
	Under binary additive valuations, all maximum Nash welfare allocations have the same utility profile. Further, any allocation with this utility profile is a maximum Nash welfare allocation.
\end{lemma}

Under binary additive valuations, given the observations above, the sets of MNW and leximin allocations can be either identical or disjoint. \citet{aziz2019almost} shows that all leximin allocations are also MNW allocations, which implies that the two sets are identical.

\begin{lemma}\label{thm:equiv}
	Under binary additive valuations, the set of maximum Nash welfare allocations coincides with the set of leximin allocations.
\end{lemma}

Henceforth, we will use the terms ``MNW allocation'' and ``leximin allocation'' interchangeably. Before we define our deterministic rule, let us define this concept formally. Fix the set of agents $\agents$ and the set of goods $\goods$. A \emph{deterministic rule} $f$ takes a valuation profile $\vals$ as input and returns an allocation $\alloc$. Note that $f$ is not allowed to return ties. We say that $f$ is EF1 (resp. PO) if it always outputs an allocation that is EF1 (resp. PO). The game-theoretic literature offers the following strong desideratum to prevent strategic manipulations by agents. 

\begin{definition}[Group strategyproofness]
    A deterministic rule $f$ is called \emph{group strategyproof} (GSP) if there do not exist valuation profiles $\vals$ and $\vals'$, and a group of agents $C \subseteq \agents$, such that $v'_k = v_k$ for all $k \in \agents \setminus C$ and $v_j(A'_j) > v_j(A_j)$ for all $j \in C$, where $\alloc = f(\vals)$ and $\alloc' = f(\vals')$.
\end{definition}

A weaker requirement, which only imposes the above property for group $C$ of size $1$ (i.e. prevents manipulations by a single agent) is commonly known as strategyproofness (SP). 
We are now ready to define our rule, which chooses a special MNW allocation. 

\begin{definition}[$\mnwt$]
    The deterministic rule $\mnwt$ returns an allocation $\alloc$ such that:
    \begin{enumerate}
        \item $\alloc$ is an MNW allocation with lexicographically greatest utility vector among all MNW allocations (i.e., among all MNW allocations, it maximizes $v_1(A_1)$, subject to that maximizes $v_2(A_2)$, and so on);\footnote{We note that tie-breaking by agent index is without loss of generality. One can break ties according to any given ordering of the agents, and the corresponding rule will still satisfy all the desiderata.} and
        \item $\alloc$ is minimally complete (i.e. $A_{\agents} = \goods \setminus \nonvalued$).
    \end{enumerate}
    If there are several allocations satisfying both conditions, $\mnwt$ arbitrarily picks one. 
\end{definition}

First, observe that $\mnwt$ is well-defined, i.e., that the set of allocations satisfying both conditions is non-empty. Indeed, the set of allocations satisfying the first condition is trivially non-empty. And for any allocation in this set, there is a corresponding minimally complete allocation\emdash obtained by throwing away all non-valued goods\emdash which has the same utility vector, and therefore still satisfies the first condition. 

The following result establishes the compelling properties of $\mnwt$. The key idea behind the polynomial-time computability is as follows. \citet{DS15b} and \citet{BKV18b} show that under binary additive valuations, \emph{an} MNW allocation can be computed efficiently. Starting from this MNW allocation, we keep moving to lexicographically better MNW allocations, as in the definition of $\mnwt$. 
The algorithm is formally presented as \Cref{alg:mnwt}. 

\begin{theorem}\label{thm:det-mnw}
    Under binary additive valuations, $\mnwt$ is envy-free up to one good, Pareto optimal, group strategyproof, and polynomial-time computable.
\end{theorem}

Before diving into the proof, we need another concept that we will use repeatedly. The \emph{graph of an allocation} $\alloc$ is a directed graph $G(\alloc) = (V,E)$, where $V$ contains a vertex for each agent, and there is a directed edge $(i,j) \in E$ if and only if there is a good in agent $j$'s bundle that agent $i$ likes (i.e., $A_j \cap V_i \neq \emptyset$). Given a path $P = (u_1,\ldots,u_k)$ in $G(\alloc)$, let $P(\alloc)$ denote an allocation obtained by transferring a good $g \in A_{u_{\ell+1}} \cap V_{u_{\ell}}$ from agent $u_{\ell+1}$ to agent $u_{\ell}$ for each $\ell \in [k-1]$; we refer to this operation as \emph{passing back along $P$}. We characterize MNW allocations in terms of non-existence of a special path in their graph.

\begin{lemma}\label{lem:passback}
	Let $\alloc$ be a Pareto optimal allocation, $P$ be a path from agent $i$ to agent $j$ in $G(\alloc)$, and $\alloc' = P(\alloc)$ be obtained by passing back along $P$. Then 
	$v_j(A'_j) = v_j(A_j)-1$, $v_i(A'_i) = v_i(A_i)+1$, and $v_k(A'_k) = v_k(A_k)$ for all $k \in \agents\setminus\set{i,j}$.
\end{lemma}
\begin{proof}
	Note that if good $g$ is being passed from agent $u_{\ell+1}$ to agent $u_{\ell}$ on path $P$, then by definition $u_{\ell}$ must like $g$. Hence, $g$ is a valued good. Thus, by PO, $u_{\ell+1}$ must like $g$ as well. Thus, each agent on $P$ except $i$ and $j$ loses a good she likes and receives a good she likes, agent $j$ only loses a good she likes, and agent $i$ only receives a good she likes. 
\end{proof}

\begin{lemma}\label{lem:path}
	A Pareto optimal allocation $\alloc$ is an MNW allocation if and only if there is no directed path from an agent $i$ to an agent $j$ in $G(\alloc)$ such that $v_j(A_j) > v_i(A_i)+1$. 
\end{lemma}
\begin{proof}
	Lemma~3 of \citet{BKV18b} establishes that $\alloc$ is an MNW allocation if and only if there is no directed path $P$ such that passing back along $P$ strictly increases Nash welfare.\footnote{Technically, either more agents receive positive utility, or the product of positive utilities increases.} Given that $\alloc$ is PO, \Cref{lem:passback} implies that this is equivalent to $(v_j(A_j)-1)\cdot (v_i(A_i)+1) > v_j(A_j) \cdot v_i(A_i)$, which is equivalent to $v_j(A_j) > v_i(A_i)+1$.
	
\end{proof}

We are now ready to prove \Cref{thm:det-mnw}.
\begin{proof}
\citet{CKMP+19} already establish that all MNW allocations are EF1 and PO, even for general additive valuations. Hence, $\mnwt$ is also trivially EF1 and PO. 

\paragraph{Group strategyproofness.} We now establish that it is GSP. Note that this holds regardless of how ties are broken among allocations satisfying the two conditions in the definition of $\mnwt$.

First, notice that if $\alloc = \mnwt(\vals)$, then $\alloc$ is minimally complete and PO. Hence, if agent $i$ receives good $g$, she must like it. In other words, $A_i \subseteq V_i$, and thus, $v_i(A_i) = |A_i|$ for each agent $i \in \agents$. Consequently, $A_U \subseteq V_U$ for every subset of agents $U \subseteq \agents$. We will use this observation repeatedly. 

Next, for an allocation $\alloc$ and agent $i \in \agents$, define $\li = \set{j \in \agents\ |\ v_j(A_j) < v_i(A_i)}$ to be the set of agents who have strictly less utility than agent $i$, and define $\si$ to be the set of agents reachable from $\li \cup \set{i}$ in $G(\alloc)$. The following lemma shows that agents in $\si$ must collectively receive all the goods that they like. 
 
 \begin{lemma}\label{lem:sisame}
    If $\alloc = \mnwt(\vals)$, then for each agent $i \in \agents$, we have $A_{\si} = V_{\si}$.
 \end{lemma}
 \begin{proof}
    We have already established that $A_{\si} \subseteq V_{\si}$. Suppose for contradiction that there exists a good $g \in V_{\si} \setminus A_{\si}$. Then, by the construction of $G(\alloc)$, there would have been an edge from an agent in $\si$ who likes $g$ to an agent outside of $\si$ who is allocated $g$ under $\alloc$ (note that $g$ is valued, so it must be allocated under $\alloc$). However, the definition of $\si$ implies that it cannot have any outgoing edges, otherwise the set of agents reachable from $\li \cup \set{i}$ could be expanded. Hence, we have $A_{\si} = V_{\si}$.
 \end{proof}

Next, we show that even though $\si$ contains all agents reachable from $\li \cup \set{i}$, an agent in $\si$ cannot have much higher utility than agent $i$ does. 
 
\begin{lemma}\label{lem:si}
    If $\alloc = \mnwt(\vals)$, then for each agent $i \in \agents$ and each agent $j \in \si$, we have that $v_j(A_j) \leq v_i(A_i) + 1$, and if $j \ge i$, then $v_j(A_j) \leq v_i(A_i)$. 
\end{lemma}
\begin{proof}
	This is trivial for $j=i$, so assume $j\neq i$, and suppose for contradiction that the statement is false. Hence, there exist agents $i \in \agents$ and $j \in \si$ such that either $v_j(A_j) \ge v_i(A_i)+2$, or $v_j(A_j) = v_i(A_i)+1$ and $j > i$. 
	
	Since $j \in \si$, there exists a path from an agent $k \in \li \cup \set{i}$ to agent $j$. Further, $k \in \li \cup \set{i}$ implies that $v_k(A_k) \le v_i(A_i)$ by definition. 
	
	Now, in the former case, we would have that there exists a path from agent $k$ to agent $j$ and $v_j(A_j) \ge v_i(A_i)+2 \ge v_k(A_k)+2$. However, this contradicts \Cref{lem:path}. 
	
	In the latter case, we consider two sub-cases. If $k \neq i$, then $k \in \li$. Hence, $v_k(A_k) < v_i(A_i)$. This implies $v_j(A_j) = v_i(A_i)+1 \ge v_k(A_k)+2$, which leads to a contradiction as pointed out above. If $k=i$, then we have a path from agent $i$ to agent $j > i$ with $v_j(A_j) = v_i(A_i)+1$. Once again, passing back along this path would result in an allocation $\alloc'$ under which $v_t(A'_t) = v_t(A_t)$ for all $t \neq i,j$, $v_i(A'_i) = v_i(A_i)+1 = v_j(A_j)$, and $v_j(A'_j) = v_j(A_j)-1 = v_i(A_i)$. Since $\alloc'$ has the same utility profile as $\alloc$, by \Cref{lem:nash-unique-profile}, $\alloc'$ is also a maximum Nash welfare allocation. Further, since a lower-indexed agent receives higher utility, $\alloc'$ is lexicographically better than $\alloc$, which contradicts the fact that $\alloc$ was returned by $\mnwt$.
\end{proof}

We are now ready to show that $\mnwt$ is GSP. Suppose for contradiction that there exist valuation profiles $\vals$ and $\vals'$, and a set of agents $C \subseteq \agents$ such that $v_t = v'_t$ for all $t \notin C$ and $v_j(\ali{j}) > v_j(\ati{j})$ for each $j \in C$, where $\alltruth = \mnwt(\vals)$ and $\alllie = \mnwt(\vals')$.

Let $i = \min \left[\argmin_{t \in C} v_t(\ati{t})\right]$ be the agent in $C$ who has the lowest index among all agents in $C$ having the minimum utility under honest reporting. For simplicity, let us denote $S = \sit$. We have that for every $j \in C$, $|V_j \cap \ati{j}| < |V_j \cap \ali{j}|$. Further, since $\ati{j} \subseteq V_j$, this simplifies to $|\ati{j}| < |V_j \cap \ali{j}|$. When $j \in S \cap C$, we get $|\ati{j}| < |V_j \cap \ali{j}| \leq |V_S \cap \ali{j}|$ because $V_j \subseteq V_S$. 

Let $R \subseteq S$ be the set of agents in $S$ from which some agent in $C$ is reachable in $G(\alllie)$. We now establish that some non-manipulating agent in $R$ must receive strictly fewer goods under $\alllie$ than under $\alltruth$. 

\begin{lemma}\label{lem:liesgetsless}
There exists $j^* \in R \setminus C$ with $|\ali{j}| < |\ati{j}|$.
\end{lemma}
\begin{proof}
Suppose for a contradiction that for all $j \in R \setminus C$, $|\ali{j}| \geq |\ati{j}|$. Take a $j \in R\setminus C$. Since $j \notin C$, she reports $v'_j = v_j$. Hence, we have $\ali{j} \subseteq V'_j = V_j$. Further, since $j \in R \subseteq S$, we have $V_j \subseteq V_S$ by definition. We conclude that for each $j \in R\setminus C$, $\ali{j} \subseteq V_S$, so $|\ali{j} \cap V_S| = |\ali{j}| \geq |\ati{j}|$. 

Additionally, for each $j \in R \cap C \subseteq C$, we have that $|\ali{j} \cap V_S| \geq |\ali{j} \cap V_j| > |\ati{j}|$. Since bundles of an allocation are disjoint, we can add these inequalities over all $j \in (R\setminus C) \cup (R \cap C) = R$ to get $|\ali{R} \cap V_S| > |\ati{R}|$. The inequality is strict because $R \cap C \neq \emptyset$ as $i \in R\cap C$ by definition. Now, recall that by \Cref{lem:sisame}, $\ati{S} = V_S$. Hence, this becomes $|\ali{R} \cap \ati{S}| > |\ati{R}|$. 

This implies that there must exist a good $g$ that is in both $\ali{R}$ and $\ati{S}$ but not in $\ati{R}$. Therefore, there exist agents $t \in R$ and $k \in S \setminus R$ such that $g \in \ali{t}$ and $g \in \ati{k}$. The latter implies $v_k(g) = 1$ due to Pareto optimality of $\alltruth$. 

Since $k \notin R$, by definition $k$ does not have a path to an agent in $C$ under $G(\alllie)$. This trivially implies $k \notin C$ since every vertex is reachable from itself. Since only members of $C$ changed their reported valuations, $v'_k(g) = v_k(g) = 1$. It follows that there must be an edge from agent $k$ to agent $t$ in $G(\alllie)$. Thus, all vertices reachable from $t$ are also reachable from $k$. But then, $t \in R$ implies $k \in R$, which is a contradiction.
\end{proof}

Consider an agent $j^* \in R \setminus C$ as per \Cref{lem:liesgetsless}. Since $j^* \in R$, there must exist a path $P$ in $G(\alllie)$ from $j^*$ to some agent $k \in C$. Let $\alloc'$ denote the allocation obtained by passing back along path $P$. We show that $\alloc'$ must be preferred to $\alllie$ by $\mnwt$ given valuation profile $\vals'$, contradicting the fact that $\mnwt(\vals') = \alllie$. 

Note that since $\alllie$ is PO under valuation profile $\vals'$, when constructing $\alloc'$ from $\alllie$, we get $v'_t(A'_t) = v'_t(\ali{t})$ for all $t \neq j^*,k$, $v'_k(A'_k) = v'_k(\ali{k})-1$, and $v_{j^*}(A'_{j^*}) = v_{j^*}(\ali{j^*})+1$ due to \Cref{lem:passback}; recall that $j^* \notin C$, so $v_{j^*} = v'_{j^*}$. Further, the set of goods allocated does not change. Hence, $\alloc'$ remains minimally complete. 

If $v_{j^*}(\ali{j^*}) + 2 \le v'_k(\ali{k})$, then, by \Cref{lem:path}, this contradicts the fact that $\alllie$ is an MNW allocation. Hence, we must have 
\begin{align}
    v_{j^*}(\ali{j^*})+1 &\ge v'_k(\ali{k}) = |\ali{k}| \ge v_k(\ali{k}) \ge v_k(\ati{k})+1 \ge v_i(\ati{i})+1,\label{eqn:1}
\end{align}
where the second transition is because $\alllie$ is minimally complete and PO under valuation profile $\vals'$, the fourth transition is because $k \in C$, and the last transition is due to the choice of $i$. On the other hand, we also have 
\begin{equation}\label{eqn:2}
v_{j^*}(\ali{j^*})+1 \le v_{j^*}(\ati{j^*}) \le v_i(\ati{i})+1,
\end{equation}
where the first transition holds because, due to \Cref{lem:liesgetsless}, $v_{j^*}(\ali{j^*}) \le |\ali{j^*}| < |\ati{j^*}| = v_{j^*}(\ati{j^*})$, and the second transition holds due to \Cref{lem:si} and the fact that $j^* \in R \subseteq S$.

Putting Equations~\eqref{eqn:1} and~\eqref{eqn:2} together, we have 
\begin{align*}
v_{j^*}(\ali{j^*}) + 1 &= v_{j^*}(\ati{j^*}) = v_i(\ati{i}) + 1 = v_k(\ati{k}) + 1 = v'_k(\ali{k}).
\end{align*}

By the second equality and \Cref{lem:si}, we must have $j^* < i$. By the third equality, the fact that $k$ and $i$ have the same utility under $\alltruth$, and the definition of $i$, we have that $k \geq i$. Therefore, $k > j^*$. Then, as argued in the proof of \Cref{lem:si}, under the valuation profile $\vals'$, $\alloc'$ has the same utility profile as $\alllie$, and thus, by \Cref{lem:nash-unique-profile}, it is an MNW allocation. Further, it is lexicographically better than $\alllie$ under $\vals'$, which contradicts the fact that $\alllie = \mnwt(\vals')$. 

\paragraph{Polynomial-time computability.} Finally, we show that $\mnwt$ can be computed in polynomial time. Fix an instance $(\agents,\goods,\vals)$. Without loss of generality, suppose there are no non-valued goods. This is because if there are any non-valued goods, we can simply remove them, and run the algorithm below on the remaining instance. 

Let $\utie$ be the utility vector that is lexicographically greatest among the utility vectors of all MNW allocations. Our goal is to compute an allocation that achieves this utility vector. Our algorithm relies on the following important lemma.

\begin{lemma}\label{lem:lex-path-mnw}
	Suppose $\alloc$ is an MNW allocation with utility vector $\uvec = (u_1,\ldots,u_n) \neq \utie$. Let $i$ be smallest index such that $u_i \neq \utiek{i}$. Then, $\utiek{i} = u_i + 1$, and there exists $j > i$ such that there is a path from $i$ to $j$ in $G(\alloc)$ and $u_j = u_i + 1$.
\end{lemma}
\begin{proof}
	Given two allocations $\vec{A}$ and $\vec{B}$, we define the transformation graph $G(\vec{A},\vec{B})$ similarly to \citet{FSVX19}. It has a vertex corresponding to each agent, and for each good $g$, there is a directed edge $(i,j)$ if $g \in A_i$, $g \in B_j$, and $i \neq j$; note that this may be a multi-graph. This edge signifies that $g$ must be passed from agent $i$ to agent $j$ in order to transform $\vec{A}$ to $\vec{B}$. Let $\cE_+ = \set{i \in \agents:|A_i| > |B_i|}$ and $\cE_- = \set{i \in \agents: |B_i| > |A_i|}$. 
	
	Corollary~3 by \citet{FSVX19} establishes that edges in $G(\vec{A},\vec{B})$ can be decomposed into a set of cycles $\mathcal{C}$ and a set of paths $\mathcal{P}$ such that each path begins at an agent in $\cE_+$ and ends at an agent in $\cE_-$. Although they do not mention this in the statement of their corollary, they in fact prove something stronger: for every agent $i \in \cE_+$ (resp. $\cE_-$), there is a path $P \in \mathcal{P}$ beginning (resp. ending) at $i$.
	
	Next, we notice that the transformation graph $G(\vec{A},\vec{B})$ is closely related to graphs $G(\vec{A})$ and $G(\vec{B})$. Specifically, if $\vec{A}$ and $\vec{B}$ are both PO, and there is an edge $(i,j)$ in $G(\vec{A},\vec{B})$, then there must be an edge $(j,i)$ in $G(\vec{A})$ and an edge $(i,j)$ in $G(\vec{B})$. To see this, note that if $(i,j)$ is an edge in $G(\vec{A},\vec{B})$, then there is a good $g \in B_j \cap A_i$. Since both allocations are PO, both $i$ and $j$ must like $g$ (i.e. $g \in V_j \cap V_i$). Now, $g \in A_i \cap V_j$ implies that edge $(j,i)$ exists in $G(\vec{A})$, and $g \in B_j \cap V_i$ implies that edge $(i,j)$ exists in $G(\vec{B})$. Extending this argument, we get that a path from $i$ to $j$ in $G(\vec{A}, \vec{B})$ implies a path from $j$ to $i$ in $G(\vec{A})$ and a path from $i$ to $j$ in $G(\vec{B})$.
	
	We let $\atie$ be some MNW allocation with utility vector $\utie$, and consider $G(\alloc, \atie)$. Since both $\vec{A}$ and $\atie$ are PO, and there are no non-valued goods, we have that $v_j(A_j) = |A_j|$ and $v_j(\atiek{j}) = |\atiek{j}|$ for all agents $j$. 
	
	Consider agent $i$ defined in the lemma statement. First, note that $u_i > \utiek{i}$ would violate lexicographic maximality of $\utie$. Hence, we must have $u_i < \utiek{i}$, i.e., $i \in \cE_{-}$. Therefore, there must exist a path $P$ from some agent $j \in \cE_+$ to agent $i$ in $G(\alloc,\atie)$. Note that this means there is a path from $j$ to $i$ in $G(\atie)$, and, crucially, a path from $i$ to $j$ in $G(\alloc)$. However, for all $j < i$, we have $u_j = \utiek{j}$, i.e., $|A_j| = |\atiek{j}|$, i.e., $j$ belongs to neither $\cE_+$ nor $\cE_-$. Hence, we must have $j > i$. 
	
	Since both $\alloc$ and $\atie$ are MNW allocations, and there is a path from $i$ to $j$ in $G(\alloc)$ and a path from $j$ to $i$ in $G(\atie)$, by \Cref{lem:path}, we have that 
	\begin{equation}\label{eqn:A-Atie-1}
	|A_j| \leq |A_i| + 1 \text{ and } |\atiek{i}| \leq |\atiek{j}| + 1.
	\end{equation} 
	In addition, we have $i \in \cE_-$ and $j \in \cE_+$. Hence,
	\begin{equation}\label{eqn:A-Atie-2}
	|A_i| < |\atiek{i}| \text{ and } |\atiek{j}| < |A_j|.
	\end{equation}
	Combining \Cref{eqn:A-Atie-1,eqn:A-Atie-2}, we have 
	\begin{equation}\label{eqn:A-Atie-3}
	|A_i| < \atiek{i} \le \atiek{j}+1 < |A_j|+1.
	\end{equation}
	Since all values are integers, we have $|A_i| + 2 \leq |A_j| + 1$, i.e., $|A_i| + 1 \leq |A_j|$. Given \Cref{eqn:A-Atie-1}, this implies $|A_j| = |A_i| + 1$. Substituting this equality in \Cref{eqn:A-Atie-3}, we also get $|\atiek{i}| = |A_i|+1$, as desired. 
\end{proof}

Suppose $\alloc$ is an MNW allocation with $\uvec \neq \utie$. Hence, \Cref{lem:lex-path-mnw} holds. Consider the agents $i,j$ and path $P$ identified in the lemma. Let $\alloc' = P(\alloc)$ have utility vector $\uvec'$. Then, by \Cref{lem:passback}, we have that $u'_j = u_j - 1 = u_i$ and $u'_i = u_i+1 = u_j = \utiek{i}$. Hence, it can be checked that $\alloc'$ is an MNW allocation, and its utility vector $\uvec'$ has a strictly longer prefix matching $\utie$ than $\uvec$ does. Consequently, if $\alloc$ is in fact an MNW allocation with $\uvec = \utie$, then no such path can exist. 

We are now ready to describe our algorithm. It starts by computing any MNW allocation $\alloc$. \citet{DS15b,BKV18b} provide efficient algorithms for computing an MNW allocation under binary additive preferences, which can be used. Then, our algorithm iteratively finds the smallest index $i$ that has a path to some $j > i$ with $|A_j| = |A_i|+1$ and passes back along such a path. By the above arguments, this must terminate in at most $n$ iterations at an MNW allocation with lexicographically greatest utility vector, which the algorithm returns. A somewhat simpler but equivalent description of the algorithm is given as \Cref{alg:mnwt}. 

\begin{algorithm}
	\caption{A polynomial-time algorithm to compute $\mnwt$ for binary additive valuations}
	\label{alg:mnwt}
	\begin{algorithmic}[1]
		\STATE Compute an MNW allocation $\alloc^{\!0}$
		\FOR{$i = 1,\ldots,n$}
		\IF{there is a path $P$ in $G(\alloc^{\!i-1})$ from agent $i$ to some agent $j > i$ with $|A^{i-1}_j| = |A^{i-1}_i| + 1$}
		\STATE $\alloc^{\!i} \gets P(\alloc^{\!i-1})$
		\ELSE
		\STATE $\alloc^{\!i} \gets \alloc^{\!i-1}$
		\ENDIF
		\ENDFOR
		\RETURN $\alloc^{\!n}$
	\end{algorithmic}
\end{algorithm}

To formally argue correctness, we use induction on $i$ and show that $\alloc^{\!i}$ is an MNW allocation with $v_k(A^i_k) = \utiek{k}$ for $k \in [i]$ (that is, its utility vector matches $\utie$ in the first $i$ components). The base case with $i=0$ is trivial as $\alloc^0$ is an MNW allocation. Suppose the induction hypothesis holds for $i-1$. 

Now, if there is a path $P$ from $i$ to some $j > i$ in $G(\alloc^{\!i-1})$ as identified in \Cref{lem:lex-path-mnw}, then we know that $v_i(A^{i-1}_i) + 1 = \utiek{i}$. However, as argued above, passing back along this path results in an MNW allocation $\alloc^{\!i}$ under which $v_i(A^i_i) = v_i(A^{i-1}_i) + 1 = \utiek{i}$. Further, it does not change the utilities to any agent $k < i$. Hence, $\alloc^{\!i}$ is an MNW allocation whose utility vector matches $\utie$ in the first $i$ components, as desired. On the other hand, if there is no such path, then by \Cref{lem:lex-path-mnw} and the induction hypothesis, it must be the case that $v_i(A^{i-1}_i) = \utiek{i}$, so setting $\alloc^{\!i} \gets \alloc^{\!i-1}$ achieves the desired goal. 

To see that the running time is polynomial, first note that we can compute an arbitrary MNW allocation in polynomial time for binary additive utilities. In each iteration of the for loop, constructing $G(\alloc^{i-1})$ and searching for a desired path in this graph can also be done in polynomial time. Finally, computing an allocation by passing back along a path can be done in polynomial time. Since the for loop runs for $n$ iterations, the overall running time is polynomial. 
\end{proof}


\section{Randomized Setting}\label{sec:randomized}

In the previous section, we established the existence of a deterministic rule which is EF1, PO, and GSP. For deterministic rules, it is necessary to relax EF to EF1. For example, in case of a single good that is liked by two agents, giving it to either agent would be EF1 but not EF. However, if one is willing to randomize, the natural solution of assigning the good to an agent chosen at random would be ``ex ante EF'' in addition to being ``ex post EF1''. This is because each deterministic allocation in the support is EF1, but in expectation, no agent envies the other. This leads to a natural question. \emph{Can randomness help achieve ex ante EF and ex post EF1, in addition to PO and GSP?} 

In this section, we answer this question affirmatively for binary additive valuations. In parallel to our work, \citet{FSV20} show that ex ante EF and ex post EF1 can be achieved simultaneously even under general additive valuations, but they show an impossibility when ex ante PO is added to the combination. Our positive result circumvents this impossibility for binary additive valuations. Additionally, it satisfies GSP, which \citeauthor{FSV20} do not consider. Let us first formally extend our framework to include randomness. 

\begin{definition}[Fractional and randomized allocations]
	A \emph{fractional allocation} $\alloc = (A_1,\ldots,A_n)$ is such that $A_i(g) \in [0,1]$ denotes the fraction of good $g$ allocated to agent $i$ and $\sum_{i \in \agents} A_i(g) \le 1$ for each good $g$. A \emph{randomized allocation} $\dist$ is a probability distribution over deterministic allocations. 
\end{definition}

There is a natural fractional allocation $\alloc$ associated with each randomized allocation $\dist$, where $A_{i}(g)$ is the probability of good $g$ being allocated to agent $i$ under $\dist$. In this case, we say that randomized allocation $\dist$ \emph{implements} fractional allocation $\alloc$. There may be several randomized allocations implementing a given fractional allocation. 

We refer to the expected utility of agent $i$ under a randomized allocation $\dist$ as simply the utility of agent $i$ under $\dist$. Note that this is equal to the utility of agent $i$ from the corresponding fractional allocation $\alloc$, defined as $v_i(A_i) = \sum_{g \in \goods} A_i(g) \cdot v_i(g)$. With this notation, the definitions of envy-freeness and Pareto optimality extend naturally to fractional allocations.\footnote{In case of Pareto optimality of a fractional allocation, we require that no other \emph{fractional} allocation Pareto-dominate it.} We say that a randomized allocation $\dist$ is ex ante envy-free (resp. ex ante Pareto optimal) if the corresponding fractional allocation $\alloc$ is envy-free (resp. Pareto optimal).

With a fixed set of agents $\agents$ and a fixed set of goods $\goods$, a \emph{randomized rule} $f$ takes a valuation profile $\vals$ as input and returns a randomized allocation $\dist$. We say that $f$ is \emph{ex ante envy-free} (resp. \emph{ex ante Pareto optimal}) if it always returns a randomized allocation that is ex ante envy free (resp. ex ante Pareto optimal). We say that $f$ is \emph{ex ante group strategyproof} if no group of agents can misreport their preferences so that each agent in the group receives strictly greater expected utility. Note that these ex ante guarantees depend only on the fractional allocation corresponding to the randomized allocation returned by $f$. Hence, when talking about ex ante guarantees, we will think of the randomized rule $f$ as directly returning a fractional allocation. However, when talking about ex post guarantees, we would need to specify which randomized allocation $f$ returns.

\begin{definition}[Ex post EF1]
        We say that a randomi	le is ex post EF1 if it always returns a randomized allocation that is ex post EF1.
\end{definition}

Fractional leximin allocations, like their deterministic counterpart, lexicographically maximize the utility profile among all fractional allocations. The same can be said about fractional MNW allocations; however, we can skip the first step of maximizing the number of agents who receive positive utility because in the fractional case we can simultaneously give positive utility to every agent who likes at least one good (and thus can possibly get positive utility). 
\begin{definition}[Fractional MNW allocations]
    We say that a fractional allocation is a fractional maximum Nash welfare allocation if it maximizes the product of utilities of agents who do not have zero value for every good. 
\end{definition}

\citet{BM04}, \citet{BMR05}, and \citet{KPS18} study fractional leximin allocations under an assignment setting, and establish several desirable properties. In addition, fractional MNW allocations, also known as competitive equilibria with equal incomes (CEEI), are widely studied in fair division with additive valuations~\cite{Var74,Orlin10,CGG13,CG18}. Our first result shows that under binary additive valuations, these two fundamental concepts coincide. 

\begin{theorem}\label{thm:rand-leximin-mnw-equiv}
	Under binary additive valuations, the set of fractional leximin allocations coincides with the set of fractional maximum Nash welfare allocations. All such allocations have identical utility vectors.
\end{theorem}
\begin{proof}
	We begin by showing that there exists a utility vector $\ulex$ (resp. $\umnw$) such that the set of fractional leximin allocations (resp. fractional MNW allocations) is exactly the set of all fractional allocations with utility vector $\ulex$ (resp. $\umnw$). In fact, this step holds even under general additive valuations. Then, for binary additive valuations, we will show that $\ulex = \umnw$, implying the desired result.
	
	Let $\alloc$ be a fractional leximin allocation with utility vector $\uvec$. Trivially, every fractional allocation with utility vector $\uvec$ is also a fractional leximin allocation. We want to show that there is no fractional leximin allocation $\alloc'$ with utility vector $\uvec' \neq \uvec$. Suppose for contradiction that there is one. Consider the fractional allocation $\alloc'' = \nicefrac{1}{2} \cdot \alloc + \nicefrac{1}{2} \cdot \alloc'$, i.e., $A^{''}_i(g) =  \nicefrac{1}{2} \cdot A_i(g) + \nicefrac{1}{2} \cdot A'_i(g)$ for each agent $i$ and good $g$. Because valuations are additive, its utility vector is $\uvec'' = \nicefrac{1}{2} \cdot \uvec + \nicefrac{1}{2} \cdot \uvec'$. 
	
	The key step to observe is that if $\uvec$ and $\uvec'$ are utility vectors of two fractional leximin allocations, and $\uvec \neq \uvec'$, then $\uvec''$ is strictly better than both $\uvec$ and $\uvec'$ according to leximin comparison, which yields the desired contradiction. To see why this is true, note that because both $\alloc$ and $\alloc'$ are fractional leximin allocations, their utility profiles must be identical; call it $\sorted^*$. Then, for any $k \in [n]$, the sum of the $k$ lowest utilities under $\alloc''$ is the average of the sum of utilities of the corresponding $k$ agents under $\alloc$ and $\alloc'$. Since each sum is at least the sum of the first $k$ components of $\sorted^*$, it follows that the sum of the $k$ lowest utilities under $\uvec''$ is at least as much as the sum of the $k$ lowest utilities under $\uvec$ or $\uvec'$, i.e., $\uvec''$ is at least as good as $\uvec$ and $\uvec'$ under leximin comparison. To see why it is strictly better, recall that $\uvec \neq \uvec'$. Let agent $i$ be such that $u_i \neq u'_i$, and among such agents, one with the lowest $u''_i$. Without loss of generality, assume $u_i < u'_i$. Then, $u_i < u''_i < u'_i$. Let $N = \set{j \in \agents : u''_j < u''_i}$ and $k = |N|$. By the definition of $N$, for each $j \in N$, we have $u_j = u'_j = u''_j$. Hence, the $k$ smallest values in $\uvec''$ also appear in $\uvec$ and the $k+1^\text{st}$ smallest value in $\uvec''$ (which is $u''_i$) is strictly higher than the $k+1^{\text{st}}$ smallest value in $\uvec$ (which is less than $u''_i$), which shows that $\uvec''$ is strictly better than $\uvec$ under leximin comparison. The comparison to $\uvec'$ follows since $\uvec$ and $\uvec'$ have identical sorted order, and thus are equivalent under leximin comparison.
	
	Thus, we have established that there exists a utility vector, say $\ulex$, such that the set of fractional leximin allocations is the set of allocations with utility vector $\ulex$. It is easy to see that the above argument holds for fractional MNW allocations as well. Crucially, the key step in the paragraph above holds because the MNW objective function (product of utilities of agents who like at least one good) is a strictly concave function. Hence, if $\uvec \neq \uvec'$ have equal objective value, then $\uvec''$ has a strictly better objective value than both of them. Let $\umnw$ denote the utility vector for fractional MNW allocations.
	
	Finally, we need to show that $\ulex = \umnw$. Fix arbitrary fractional leximin and fractional MNW allocations $\alloc^{\lex}$ and $\alloc^{\mnw}$. Suppose this is not true. Let $\wlex$ and $\wmnw$ be the utility profiles corresponding to $\ulex$ and $\umnw$, respectively. Let $k$ be the smallest index such that $w^{\lex}_k \neq w^{\mnw}_k$. Because $\wlex$ is the leximin-optimal utility profile, we must have $w^{\lex}_k > w^{\mnw}_k$. Because fractional leximin and fractional MNW allocations are PO, they have identical sum of utilities. Hence, there exists $t$ such that $w^{\lex}_t < w^{\mnw}_t$. Choose the smallest such $t$. Then, we have that for each $k < t$, $w^{\lex}_k \ge w^{\mnw}_k$, and for at least one $k < t$, $w^{\lex}_k > w^{\mnw}_k$. Thus, $\sum_{k=1}^{t-1} w^{\lex}_k > \sum_{k=1}^{t-1} w^{\mnw}_k$. It is also worth noting that $w^{\mnw}_{t-1} \le w^{\lex}_{t-1} \le w^{\lex}_t < w^{\mnw}_t$. 
	
	Let $N$ denote the set of agents with the $t-1$ lowest utilities under $\umnw$. Since the collective utility these agents receive is at the minimum the sum of the first $t-1$ values of $\wlex$, they receive strictly higher total utility under $\ulex$ than under $\umnw$. Thus, there must exist an agent $i \in N$ and a good $g \in V_i$ such that a positive fraction of $g$ is allocated to an agent $j \notin N$ under $\alloc^{\mnw}$. Since $u^{\mnw}_j \ge w^{\mnw}_t > w^{\mnw}_{t-1} \ge u^{\mnw}_i$, it follows that transferring a small enough fraction of good $g$ from agent $j$ to agent $i$ in $\alloc^{\mnw}$ strictly improves the Nash welfare, which is a contradiction. Hence, $\ulex = \umnw$. Consequently, the set of fractional leximin and fractional MNW allocations coincide, and these allocations have identical utility vectors.
\end{proof}
Note that the identical utility vector guarantee in \Cref{thm:rand-leximin-mnw-equiv} is much stronger than the identical utility profile guarantee in the deterministic case (\Cref{lem:leximin-unique-profile,lem:nash-unique-profile}).  

Even under general additive valuations, it is known that every fractional MNW allocation is ex ante EF and ex ante PO~\cite{Var74}, and one such allocation can be computed in strongly polynomial time~\cite{Orlin10,Vegh13}. Hence, these properties carry over to our binary additive valuations domain, and due to \Cref{thm:rand-leximin-mnw-equiv}, also apply to fractional leximin allocations.

For ex ante GSP, we build on the literature on fractional leximin allocations. \citet{KPS18} show that returning a fractional leximin allocation satisfies ex ante EF, ex ante PO, and ex ante GSP whenever four key requirements are satisfied. We describe them below, and show that they are easily satisfied under binary additive valuations, if we return a minimally complete leximin allocation. Hence, we define our fractional leximin/MNW rule to always return a minimally complete fractional leximin/MNW allocation (like our deterministic rule $\mnwt$). The proof of the next result is straightforward. 

\begin{definition}[Fractional maximum Nash welfare rule]
    The fractional maximum Nash welfare rule returns a minimally complete fractional maximum Nash welfare allocation.
\end{definition}

\begin{theorem}\label{thm:rand-leximin-mnw-properties}
    Under binary additive valuations, every fractional maximum Nash welfare (equivalently, leximin) allocation is ex ante envy-free and ex ante Pareto optimal. Further, the fractional maximum Nash welfare rule is ex ante group strategyproof.
\end{theorem}
\begin{proof}
	We show that binary additive valuations satisfy the four requirements laid out by \citet{KPS18} in Section~3.2 of their paper. Using their notation, $\mathcal{A}$ denotes the set of feasible allocations, $\mathcal{P}$ denotes the set of possible preferences (i.e. weak order over $\mathcal{A}$) that agents may have, and $\mathcal{U}$ denotes the set of valuation functions that the agents may have. In our setting, $\mathcal{A}$ is the set of all fractional allocations, and each preference in $\mathcal{P}$ has a unique valuation function in $\mathcal{U}$ consistent with it, just the natural one which assigns value $1$ to each good $g$ for which the agent strictly prefers the allocation $\set{g}$ to $\emptyset$, and $0$ to every other good. We now specify the four requirements laid out by \citet{KPS18} and argue why they are satisfied in our domain. 
	
	\begin{enumerate}
		\item \emph{Convexity}. Given two feasible allocations $\alloc,\alloc' \in \mathcal{A}$ and $\lambda \in [0,1]$, we need to show there exists an allocation $\alloc'' \in \mathcal{A}$ such that $u_i(A''_i) = \lambda u_i(A_i) + (1-\lambda) u_i(A'_i)$ for each agent $i$. We can simply let $\alloc''$ to be the fractional allocation induced by the randomized allocation that selects $\alloc$ with probability $\lambda$ and $\alloc'$ with probability $1-\lambda$. 
		
		\item \emph{Equality}. \citet{KPS18} only use this property to achieve a guarantee known as proportionality. While in our domain envy-freeness implies proportionality under a complete allocation, their result requires the equality requirement to guarantee proportionality even in domains where it is not implied by envy-freeness. Hence, we do not need to show this requirement in our domain. 
		
		\item \emph{Shifting Allocations}. Given an allocation $\alloc \in \mathcal{A}$ and agents $i,j \in \agents$, we need to show there exists an allocation $\alloc' \in \mathcal{A}$ such that $v_k(A'_k) = v_k(A_k)$ for all agents $k \in \agents \setminus \set{i,j}$ and $v_i(A'_i) \geq v_i(A_j)$. For this, we can choose $\alloc'$ such that $A'_k = A_k$ for all $k \in \agents \setminus \set{i,j}$, $A'_i = A_i \cup A_j$, and $A'_j = \emptyset$. 
		
		\item \emph{Optimal Utilization}. This requires that if $\alloc \in \mathcal{A}$ is returned by the fractional MNW rule in a given instance, then for any valuation function $v \in \mathcal{U}$, $v_i(A_i) \geq v(A_i)$. Crucially, because our rule outputs a minimally complete fractional MNW allocation, agent $i$ must like each good that she is assigned a positive fraction of. Hence, $v_i(A_i) = \sum_{g \in \goods} A_i(g) \cdot 1 \ge v(A_i)$ for any binary additive valuation function $v$.
	\end{enumerate}
	
	Hence, it follows from the result of \citet{KPS18} that fractional MNW rule is ex ante envy-free, ex ante Pareto optimal, and ex ante group strategyproof.
\end{proof}

The only missing property at this point is ex post EF1. Therefore, the main question we seek to answer in this section is the following: \emph{Can every fractional MNW allocation be implemented as a distribution over deterministic EF1 allocations?} We go one step further and show that it can in fact be implemented as a distribution over deterministic MNW allocations, which are in turn EF1. Our main tool is the bihierarchy framework introduced by \citet{BCKM13}, which is a generalization of the classic Birkhoff-von Neumann theorem~\cite{Birk46,Neu53}. At a high level, the framework allows implementing any fractional allocation $\alloc$ using deterministic allocations which satisfy a set of constraints, as long as the set of constraints forms a bihierarchy structure and the fractional allocation itself satisfies those constraints. 

In our case, we start with a minimally complete fractional MNW allocation $\allocs$. Let $u^*_i$ denote the utility to agent $i$ under this allocation. We want to implement this as a randomized allocation. We impose the following constraints on a deterministic allocation $\alloc$ in the support, where $\alloc$ is represented as a matrix in which $A_i(g) \in \set{0,1}$ indicates whether good $g$ is allocated to agent $i$.
\begin{equation}\label{eqn:bihierarchy}
\begin{aligned}
    \mathcal{H}_1 : &\textstyle\sum_{i \in \agents} A_i(g) = \sum_{i \in \agents} A^*_i(g), \forall g \in \goods,\\\medskip
    \mathcal{H}_2 : &\textstyle\floor{u^*_i} \le \sum_{g \in \goods} A_i(g)\cdot v_i(g) \le \ceil{u^*_i}, \forall i \in \agents.
\end{aligned}
\end{equation}

The first family of constraints ensures that under each deterministic allocation $\alloc$, the set of goods allocated matches that under $\allocs$. Since $\allocs$ is minimally complete, this implies that $\alloc$ must be minimally complete as well. Crucially, the second family of constraints ensures that each agent has utility that is either the floor or the ceiling of her utility under $\allocs$. That is, $\alloc$ is not allowed to stray far from $\allocs$. 

It can be checked that these constraints form a bihierarchy (each of $\mathcal{H}_1$ and $\mathcal{H}_2$ is a hierarchy); for a formal definition of a hierarchy, we refer the reader to the work of \citet{BCKM13}. Importantly, they also provide a polynomial-time algorithm that computes a random allocation such that (a) it implements the fractional allocation $\allocs$, and (b) each deterministic allocation $\alloc$ in its support satisfies the constraints in \Cref{eqn:bihierarchy}. We show that in this case, every deterministic allocation in the support must be a deterministic MNW allocation, yielding the desired result. 

\begin{theorem}\label{thm:rand}
	Under binary additive valuations, given any fractional maximum Nash welfare allocation, one can compute, in polynomial time, a randomized allocation which implements it and has only deterministic maximum Nash welfare allocations in its support. 
\end{theorem}
\begin{proof}
	Let $\allocs$ be a given fractional MNW allocation with utility vector $\us$. Let $\bar{\alloc}$ be the randomized allocation implementing $\allocs$ that is returned by the polynomial-time algorithm of \citet{BCKM13} with the bihierarchy constraints in \Cref{eqn:bihierarchy}. Let $\mathcal{A}$ denote the set of deterministic allocations in the support of $\bar{\alloc}$. Our goal is to show that every allocation in $\mathcal{A}$ is an MNW allocation. 
	
	First, let us partition the set of agents $\agents$ into sets $S_1,\ldots,S_t$ such that any two agents $i$ and $j$ are in the same set if and only if $\floor{u^*_i} = \floor{u^*_j}$. For $k \in [t]$, let $L_k$ denote the common floor of utilities of agents in $S_k$ under $\allocs$, and $U_k = L_k+1$. Hence, for $k \in [t]$ and each agent $i \in S_k$, $u^*_i \in [L_k,U_k)$. Further, order the sets so that $U_k \le L_{k+1}$ for each $k \in [t-1]$. This ensures that if $i \in S_r$, $j \in S_{r'}$, and $r' > r$, then $u^*_j > u^*_i$. 
	
	We argue that for each $k \in [t]$, the agents in $\cup_{r \in [k]} S_r$ must be fully allocated all of the goods that they like (i.e. all the goods in $V_{\cup_{r \in [k]} S_r}$) under $\allocs$, resulting in $\sum_{r \in [k]} \sum_{i \in S_r} u^*_i = |V_{\cup_{r \in [k]} S_r}|$. If this is not true, then a positive fraction of some good $g \in V_{\cup_{r \in [k]} S_r}$ must be allocated to an agent $j \in S_{r'}$ for $r' > k$. Let $i \in \cup_{r \in [k]} S_r$ be an agent such that $g \in V_i$. Let $r \in [k]$ be such that $i \in S_r$. Then, by the above argument, we know that $u^*_j > u^*_i$. However, then, transferring a sufficiently small fraction of $g$ from agent $j$ to agent $i$ in $\allocs$ will improve the Nash welfare, which contradicts the fact that $\allocs$ is a fractional MNW allocation. 
	
	Note that in any deterministic allocation $\alloc$, $|V_{\cup_{r \in [k]} S_r}|$ is the highest utility that agents in $\cup_{r \in [k]} S_r$ can collectively have; hence, in any feasible utility vector $\uvec$, 
	\begin{equation}\label{eqn:prefix-inequality}
	\textstyle\sum_{r \in [k]} \sum_{i \in S_r} u_i \le \sum_{r \in [k]} \sum_{i \in S_r} u^*_i, \forall k \in [t].
	\end{equation}
	Because a convex combination of allocations in $\mathcal{A}$ yields the allocation $\allocs$, and utilities are additive, a convex combination of their utility vectors yields the utility vector $\us$. Hence, for the utility vector $\uvec$  of any allocation in $\mathcal{A}$, \Cref{eqn:prefix-inequality} must hold with equality. Further, by subtracting each equation from the next, we get that it must further satisfy the following. Here, $\mathcal{H}_2$ is from the bihierarchy constraints (\Cref{eqn:bihierarchy}).
	\begin{equation}\label{eqn:rounded}
	\begin{aligned}
	&\mathcal{H}_2: \floor{u^*_i} \le u_i \le \ceil{u^*_i}, \forall i \in \agents,\\
	&\mathcal{H}_3: \textstyle\sum_{i \in S_k} u_i = \sum_{i \in S_k} u^*_i, \forall k \in [t].
	\end{aligned}
	\end{equation}
	
	We say that a utility vector is a \emph{rounded} if it satisfies the constraints in \Cref{eqn:rounded}, and say that a deterministic allocation is \emph{rounded} if it has a rounded utility vector. We have already established that every allocation in $\mathcal{A}$ is a rounded allocation. The following lemma completes the proof.
	
	\begin{lemma}\label{lem:rounded-leximin}
		The set of rounded allocations coincides with the set of maximum Nash welfare allocations. 
	\end{lemma}
	\begin{proof}
		Because leximin and MNW are equivalent concepts for deterministic allocations (\Cref{thm:equiv}), we will refer to MNW allocations as leximin allocations in this proof. To establish the desired result, it is sufficient to show that given a rounded allocation $\allocround$ and an arbitrary allocation $\alloc$, $\allocround$ weakly leximin-dominates $\alloc$, and if $\alloc$ is not rounded, then $\allocround$ strictly leximin-dominates $\alloc$. Let $\uround$ and $\uvec$ be the utility vectors of $\allocround$ and $\alloc$, respectively. 
		
		First, assume that $\alloc$ is rounded. Then, both $\uround$ and $\uvec$ satisfy \Cref{eqn:rounded}. However, it is easy to see that any two utility vectors satisfying \Cref{eqn:rounded} induce the same utility profile, and thus $\uround$ trivially weakly leximin-dominates $\uvec$. To see this, note that for each $k \in [t]$, the sum of utilities of agents in $S_k$ is fixed due to $\mathcal{H}_3$. And further, for each agent $i \in S_k$, $\mathcal{H}_2$ implies that either $u_i = u^*_i = L_k$ if $u^*_i = L_k$, or $u_i \in \set{L_k,L_k+1}$ if $u^*_i \in (L_k,L_k+1)$. Thus, together, $\mathcal{H}_2$ and $\mathcal{H}_3$ fix the number of agents in $S_k$ that have utility $L_k$ and those that have utility $L_k+1$. Thus, any two utility vectors satisfying \Cref{eqn:rounded} induce identical utility vectors.
		
		Next, assume that $\alloc$ is not rounded, i.e., $\uvec$ violates either $\mathcal{H}_2$ or $\mathcal{H}_3$. Let $k \in [t]$ be the smallest index such that either $\mathcal{H}_2$ is violated for some agent $i \in S_k$, or $\mathcal{H}_3$ is violated for $S_k$. Then, by the above argument, the partial utility vectors $(u^{\round}_i)_{i \in \cup_{r < k} S_r}$ and $(u_i)_{i \in \cup_{r < k} S_r}$ induce identical utility profiles, and therefore, are leximin-equivalent. 
				
		Suppose $\mathcal{H}_3$ is violated for $S_k$. Then, because $\uvec$ satisfies \Cref{eqn:prefix-inequality}, and $\uvec$ and $\uround$ match on the total utility of agents in $S_r$ for $r < k$, we have $\sum_{i \in S_k} u_i < \sum_{i \in S_k} u^*_i = \sum_{i \in S_k} u^{\round}_i$. Because $(u^{\round}_i)_{i \in S_k}$ has a higher sum than $(u_i)_{i \in S_k}$, and because it distributes that higher sum as equally as possible, $(u^{\round}_i)_{i \in S_k}$ strictly leximin-dominates $(u_i)_{i \in S_k}$. Hence, $(u^{\round}_i)_{i \in \cup_{r \le k} S_r}$ also strictly leximin-dominates $(u_i)_{i \in \cup_{r \le k} S_r}$. To argue that $\uround$ strictly leximin-dominates $\uvec$, we need to argue that adding the remaining utilities does not change the comparison. To that end, note that for any $r' > k$, $r \le k$, $i' \in S_{r'}$ and $i \in S_r$, we have $u^{\round}_{i'} \ge L_{r'} \ge U_r \ge u^{\round}_i$. Thus, because we are only adding utilities to $\uround$ that are at least as high as the ones already added, the strict leximin-dominance continues to hold. 
		
		Next, $\mathcal{H}_2$ is violated for some agent $i \in S_k$. Then, $u_i \notin \set{\floor{u^*_i},\ceil{u^*_i}}$. By the above argument, $(u^{\round}_i)_{i \in \cup_{r < k} S_r}$ and $(u_i)_{i \in \cup_{r < k} S_r}$ induce the same utility profile. If $u_i < \floor{u^*_i}$, then $\uvec$ has strictly more agents with utility less than $L_k$, which implies that it is strictly leximin-dominated by $\uround$. If $u_i > \ceil{u^*_i}$, then $\uvec$ has either an agent $j \in S_k$ with $u_j < \floor{u^*_j} = L_k$ or more agents with utility $L_k$ than $\uround$ does. In either case, it is again easy to see that $\uround$ strictly leximin-dominates $\uvec$.

	\end{proof}
	This completes the proof of \Cref{thm:rand}.
\end{proof}

Let us amend the definition of the fractional MNW rule so that it uses \Cref{thm:rand} to implement a minimally complete fractional MNW allocation. Then, we have the following. 

\begin{corollary}\label{cor:rand-magical-rule}
Under binary additive valuations, the fractional maximum Nash welfare rule is ex ante envy-free, ex ante Pareto optimal, ex ante group strategyproof, ex post envy-free up to one good, and polynomial-time computable. 
\end{corollary}


\section{Discussion}
\label{sec:disc}

To recap, we showed that under binary additive valuations a deterministic variant of the maximum Nash welfare rule is envy-free up to one good (EF1), Pareto optimal (PO), and group strategyproof (GSP). We also demonstrated that its randomized variant is ex ante EF, ex ante PO, ex ante GSP, and ex post EF1. All our rules are polynomial-time computable. 

\citet{ABCM17} show that under general additive valuations, there is no deterministic rule that is envy-free up to one good (EF1) and strategyproof, even with two agents and $m \ge 5$ goods. At first glance, \Cref{thm:det-mnw}, which establishes $\mnwt$ as both GSP and EF1, seems to show that this impossibility result does not hold for the special case of binary additive valuations. However, the impossibility result of \citet{ABCM17} only applies to rules that allocate all the goods; by contrast, $\mnwt$ does not allocate non-valued goods. This begs the following question: \emph{Under binary additive valuations, is there a deterministic rule that allocates all the goods and achieves EF1, PO, and GSP?} In the appendix, we show that this cannot be achieved by any variant of MNW. 

Another open question is whether the ex ante GSP guarantee of \Cref{cor:rand-magical-rule} can be strengthened to ex post GSP, which would require the randomized rule to be implementable as a probability distribution over deterministic GSP rules. 

Modulo these minor caveats, though, our results are the strongest one could possibly hope for in the domain of binary additive valuations.

\bibliographystyle{named}
\bibliography{abb,ultimate}

\appendix
\section*{Appendix}

\section{Example Illustrating $\mnwt$}

The following examples illustrate how our deterministic rule $\mnwt$ works.

\begin{example}
Let us denote a valuation profile by a matrix, where $n$ rows represent agents $1,\ldots,n$, $m$ columns represent goods $g_1,\ldots,g_m$, and the entry in row $i$ and column $j$ is $v_i(g_j)$. Consider the valuation profile 
$$\vals = \begin{pmatrix}
    1 & 0 \\
    1 & 0
\end{pmatrix}.$$ In this case, the unique allocation $\alloc$ that $\mnwt$ can return is given by $A_1 = \set{g_1}$ and $A_2 = \emptyset$. This is because $g_2$ cannot be allocated as it is non-valued, and $\mnwt$ must prefer the allocation which gives $g_1$ to agent $1$ over the one which gives it to agent $2$. 

Generally, however, there could be multiple allocations that $\mnwt$ can arbitrarily choose from. For example, consider the valuation profile 
$$\vals = \begin{pmatrix}
    1 & 1 & 1\\
    1 & 1 & 1
\end{pmatrix}.$$ $\mnwt$ may return any allocation which gives two goods to agent $1$ and one good to agent $2$. 
\end{example}


\section{Allocating Non-valued Goods}\label{sec:complete}

\begin{theorem}
    No deterministic rule can always output an MNW allocation, always allocate all goods, and be SP for two agents and six goods.
\end{theorem}
\begin{proof}
    Suppose there are two agents $\agents = \set{1,2}$ and six goods $\goods = \set{g_1,g_2,g_3,g_4,g_5,g_6}$
    Suppose for a contradiction there existed such a rule $F$. Consider the valuation profile $\vals^1$:
    $$\begin{pmatrix}
    1 & 1 & 0 & 0 & 0 & 0\\
    1 & 1 & 0 & 0 & 0 & 0
    \end{pmatrix}$$
    The only possible MNW allocations are those where one agent gets good $1$ and the other gets good $2$. Once this condition is met, any allocation of the non-valued goods is an MNW allocation. Since there are four such goods, there must be an agent that recieves at least two of them. Without loss of generality, agent $1$ receives goods $3$ and $4$ along with one of items $1$ and $2$, say item $1$. That is, $F(\vals^1) = \alloc^1$ such that $\set{g_1, g_3, g_4} \subseteq A^1_1$ (it is possible they received more of the nonvalued ones).
    
    Now consider another valuation profile $\vals^2$:
    $$\begin{pmatrix}
    1 & 1 & 1 & 1 & 0 & 0\\
    1 & 1 & 0 & 0 & 0 & 0
    \end{pmatrix}$$
    The only possible MNW allocations are those where agent $1$ receives goods $3$ and $4$ and agent $2$ receives goods $1$ and $2$. Therefore, regardless of the allocation chosen, the utility to agent $1$ is exactly $2$, that is if $F(\vals^2) = \alloc^2$, then $v^2_1(\alloc^2_1) = 2$. However, if agent $1$ misreports to match $\vals^1_1$, then their utility $v^2_1(\alloc^1_1) \geq 3$, as agent $1$ likes all three of $g_1,g_3$, and $g_4$. Therefore, $F$ is not SP, a contradiction.
\end{proof}

We wrote an integer linear program (ILP) to check whether there exists a full allocation for every possible binary additive valuation profile with two agents and six goods such that the resulting rule is EF1, PO, and SP. We solved the ILP using CPLEX, and determined that there indeed exists such a rule. Our program does not terminate in reasonable time when run on two agents and seven goods, so whether there exists a deterministic rule that allocates all goods and is EF1, PO, and SP is an open question for $n=2$ and $m \ge 7$. 


\end{document}